\documentclass[journal]{IEEEtran}
\ifCLASSINFOpdf
\else
\fi
\usepackage{amsmath}
\usepackage{mathtools}
\usepackage{amssymb}
\usepackage{upgreek}
\newtheorem{thm}{Theorem}[section]
\newtheorem{lem}[thm]{Lemma}
\newtheorem{col}[thm]{Corollary}
\newtheorem{defn}[thm]{Definition}
\newtheorem{rem}[thm]{Remark}
\newenvironment{proof}{{\it Proof:}}{\hfill$\square$}

\begin{document}
%

\title{Attitude Tracking for Rigid Bodies Using Vector and Biased Gyro Measurements$^*$}


%
%

\author{Eduardo~Esp\' indola~and~Yu~Tang
\thanks{$^*$ This work was supported in part by CONACyT under grant 253677 and PAPIIT-UNAM IN112421, and carried out in the National Laboratory of Automotive and Aerospace Engineering LN-INGEA.}
\thanks{E. Esp\' indola (e-mail: eespindola@comunidad.unam.mx) and Y. Tang (corresponding author, e-mail: tang@unam.mx) are with the Faculty of Engineering, National Autonomous University of Mexico, Mexico City, 04510 M\' exico. }}

%
%

\markboth{Journal of \LaTeX\ Class Files,~Vol.~XX, No.~X, XXXXXX~2022}%
{Esp\' indola \MakeLowercase{\textit{et al.}}: Adaptive  Attitude Tracking for Spacecraft Using Vector Measurements}
%



\maketitle

\begin{abstract}
The rigid-body attitude tracking using vector and biased gyro measurements with unknown inertia matrix is studied in this note.  First, a gyro-bias observer with global exponential stability is designed. Then, an attitude tracking controller based on this observer is devised, ensuring almost global asymptotic stability and almost semiglobal exponential stability. 
A separation property of the combined observer-controller is proved. 
Lastly, an adaptive attitude tracking controller relying on a modified gyro-bias observer and with no over-parametrization is developed to deal with the unknown inertia matrix.  The proposed control schemes require neither an explicit attitude representation nor any attitude estimation, but only the measurement of at least two non-collinear known inertial reference vectors and biased gyro rate, which can be obtained by common low-cost IMU sensors. Simulations are included to illustrate the proposed adaptive controller under noisy measurements.
\end{abstract}


%
\IEEEpeerreviewmaketitle

\section{Introduction}
%
%
%
%
\IEEEPARstart{T}{he} attitude control for rigid bodies has been studied for decades under the assumption that attitude  measurements are  available \cite{wen1991attitude,tayebi2008unit,mayhew2011quaternion}.
In practice, the attitude is commonly estimated from vector measurements provided by inertial measurement units (IMUs) (see, e.g.,  \cite{tayebi2006attitude,mahony2008nonlinear,tong2020passive,hamrah2021finite}). However, use of attitude estimation in closed loop with a controller must be analyzed to ensure overall stability.  
Several approaches have been proposed to address attitude control using vector measurements since the pioneering work \cite{pounds2007attitude}, which can be categorized as direct and indirect. In the direct approach, no representation of the attitude and therefore no estimation of the attitude is needed, thus, inertial vector measurements are used directly in the control synthesis \cite{pounds2007attitude,forbes2013passivity,thakur2015gyro}. In the indirect approach, on the other hand, the attitude is estimated for control purposes \cite{khosravian2011rigid,mercker2011rigid,tayebi2013inertial,benallegue2018adaptive}. A combined approach is presented in \cite{lee2015global}. The direct approach  has  advantages that it simplifies the stability analysis, and most importantly reduces the undesired equilibria introduced by the attitude observer in the closed loop. For the attitude stabilization problem, vector measurements are sufficient  \cite{pounds2007attitude,tayebi2013inertial}. When the attitude tracking is considered, gyro-rate measurements are also needed \cite{forbes2013passivity,lee2015global}, and the gyro bias must be corrected \cite{mercker2011rigid,benallegue2018adaptive}. The problem becomes  even more complicated if in addition  the inertia matrix is unknown, thus the gyro bias and the inertia matrix must be estimated simultaneously, arising a nonlinear parameterization issue. To overcome the nonlinear parameterization, over-parameterization is often employed (see, e.g., \cite{mercker2011rigid,benallegue2018adaptive} and the references cited therein). 

The rigid-body attitude tracking using vector and biased gyro measurements with unknown inertia matrix is studied in this note. The organization and contributions are stated as follows. A globally exponentially stable gyro-bias observer is designed in Section \ref{Sec:Obs}. In Section \ref{Sec:CtrlGBEst},  two attitude tracking controllers, nonadaptive and adaptive,  using vector and biased gyro measurements, are designed using the direct approach. The attitude tracking problem is formulated as an alignment problem between the vector measurements and their desired values, with the alignment errors defined by the inner product and the cross-product between them. The closed-loop system contains, in addition to a set of undesired equilibria, the desired attitude as its equilibrium. By exploiting the relationship between these error variables near the undesired equilibria, a strict Lyapunov function is proposed, which allows achieving almost global asymptotic stability (AGAS) and almost semiglobal exponential stability  (ASGES) of the desired equilibrium in contrast to only AGAS in \cite{pounds2007attitude,mercker2011rigid,tayebi2013inertial,forbes2013passivity}. Moreover, the proposed gyro-bias observer allows to prove a separation property in the combined observer and the nonadaptive controller. When the inertia matrix is unknown, a modified version of the gyro-bias observer avoids over-parametrization in the adaptive controller, in contrast with those adaptive controllers in  \cite{mercker2011rigid,benallegue2018adaptive}. In Section \ref{Sec:Sim}, the proposed adaptive controller is simulated under noisy measurements. Finally, Section \ref{Sec:Conc} draws the conclusions. Appendix A gives the proof of some technical results. 

\section{Preliminaries}\label{Sec:Prel}

\subsection{Notations}
The vector norm is denoted by $\| x\|$ $=$ $(x^T x)^{1/2}$ for $x\in\mathbb{R}^n$, while for a matrix $A\in \mathbb{R}^{n\times n}$, its norm is  $\|  A\|=\lambda_{\max}^{1/2}(A^T A)$. If in addition $A$ is symmetric, then $A=A^{T}>0$ indicates a positive definite matrix with $A^T$ denoting its transpose, and  $\lambda_{\max}(A)$ and $\lambda_{\min}(A)$ are its maximum and minimum eigenvalues, respectively. $I_n$ denotes the identity matrix of $n\times n$, $0_{n\times m}$ a zero matrix of $n\times m$. $S(\cdot) \in \mathbb{R}^{3\times 3}$ is the cross-product operator $S(u)v=u\times v$,  which satisfies $S(u)v=-S(v)u$, $S(u)S(v) - S(v)S(u) = S(S(u)v)$, and $\|S(u)v\|\leq \|u\|\|v\|$, $\forall u,v\in\mathbb{R}^{3}$. 
A unit sphere of dimension $n-1$ embedded in $\mathbb{R}^{n}$ is denoted by $\mathcal{S}^{n-1}= \{ x\in\mathbb{R}^{n} | x^{T}x= 1 \}$. A closed unit ball is expressed as $\mathcal{B}^{n-1}= \{ x\in\mathbb{R}^{n} | x^{T}x\leq 1 \}$. 

\subsection{Rotational Kinematics and Dynamics}
The attitude of a rigid body is represented by a rotation matrix $R\in SO(3)$ that transforms the body-fixed frame into an inertial frame. Commonly, unit quaternions $q=\left[ q_0 , \ q^{T}_{v}\right]^{T}\in \mathcal{S}^{3}$ are employed to represent the attitude, where $q_{0} \in [-1,1]$ and $ q_{v} \in \mathcal{B}^{2}$ are the scalar part and vector part of the quaternion, respectively, and $\mathcal S^3$ denotes the quaternion group with $\hat 1 =\left[ 1,0_{1\times 3}\right]^{T}$ as the identity quaternion.
A rotation matrix $R$ is related to a unit quaternion $q$ through  Rodriguez formula
\begin{equation}\label{eq:Rodrigues}
R(q)=I_3 + 2q_{0}S(q_{v}) + 2S^{2}(q_{v}). 
\end{equation}

The rotational kinematics and dynamics of a rigid body are
\begin{equation}\label{eq:KinematicsR}
\dot{R} = RS\left( \omega \right),
\end{equation}
\begin{equation}\label{eq:DynamicsB}
M\dot{\omega} = S(M\omega)\omega + \tau ,
\end{equation}
where $\omega\in\mathbb{R}^{3}$ and $\tau\in\mathbb{R}^{3}$ are the angular velocity and the applied torque, respectively. $M=[m_{ij}]_{i,j=1}^{3,3}=M^T>0 $ is the inertia matrix
with lower and upper bands  $\underline{m}I_{3} \leq M \leq \overline{m}I_{3}$, where $\overline{m}\vcentcolon = \lambda_{\max}(M) $, and $\underline{m}\vcentcolon = \lambda_{\min}(M) $. All are expressed in the body frame.

\subsection{Measurements}
It is assumed that $n\geq 2$ reference vectors in the inertia frame $r_i\in\mathcal{S}^{2}$, $i=1,2,...,n$ are known and their measurements in the body frame $v_{i}\in\mathcal{S}^{2}$, are accessible for all $t\geq 0$. These unit vectors are related through
\begin{eqnarray}\label{eq:vi}
v_i = R^{T} r_{i}, \quad i=1,2,\cdots , n. 
\end{eqnarray}
Furthermore, a biased gyro-rate measurement $\omega_g$ 
\begin{equation}\label{eq:VelM}
 \omega_{g} = \omega + b,
\end{equation}
is available, where $b\in\mathbb{R}^{3}$ is the constant gyro bias.

The following assumptions are made.
\begin{itemize}
\item Assumption A1: Among the $n$ inertial reference vectors $r_i$, there are at least two noncollinear vectors. 
\item Assumption A2: Gyro bias is a constant unknown vector with a known bound, i.e.,  $\|b\|\leq \theta_{b}$, with $\theta_b$ known. 
\end{itemize}

Note that, for each pair of noncollinear vector measurements $v_{1},v_{2}\in\mathcal{S}^{2}$, a third virtual vector measurement $v_{3}=S(v_{1})v_{2}/\|S(v_{1})v_{2}\|$ can be obtained, which satisfies \eqref{eq:vi} with $r_3=S(r_1)r_2/\|S(r_{1})r_{2}\|$. 

\section{Gyro Bias Observer}\label{Sec:Obs}
The estimated angular velocity $\hat \omega:=\omega_g-\hat b$ can be obtained by the gyro-rate measurement $\omega_g$ and the gyro-bias estimate $\hat b$, which is given by the following  observer
\begin{align}
\dot{\bar{b}}& = K_{f}\hat{\omega}+\gamma_{f}\sum_{i=1}^{n}k_{i}S(\Lambda_{i}v_{i})\big(v_i-v_{f_i}\big) ,  \label{eq:bb} \\
\hat{b} &= \bar{b} - \sum_{i=1}^{n}k_{i}S^{T}(v_{f_i})\Lambda_{i}v_{i} , \label{eq:be}   
\end{align}
where $0<\Lambda_{i}=\Lambda^{T}_{i}\in\mathbb{R}^{3\times 3}$ is a constant matrix gain,  $k_{i}>0$ is  the  weight assigned to each vector measurement according to its confidence level, $K_{f}: = \sum_{i=1}^{n}k_{i} S^{T}(v_{f_i})\Lambda_i S(v_{i})$, and $v_{f_i}$ is the filtered $v_i$ 
\begin{equation}\label{eq:vfp}
\dot{v}_{f_i} = \gamma_f\left( v_{i} - v_{f_i} \right),\; v_{f_i}(0)=v_{i}(0),\; \forall i=1,2,\ldots , n,
\end{equation}
with $\gamma_f>0$ the filter gain.

The following technical lemma  \cite{espindola2021global} is needed for the design of the bias observer.
\begin{lem}\label{lem1}
\emph{\textbf{(Linear filter):}} Consider the linear filter \eqref{eq:vfp}. Then, $\forall \epsilon_{f} >0$, there exists a $\underline{\gamma}_f: = \omega_{\max} / \epsilon_{f}$, being $\omega_{\max}<\infty$  the bandwidth of the overall system,  such that $\|v_{i} - v_{f_i} \| < \epsilon_{f}$ for $i=1,2,...,n$ and $t\geq 0$ provided that $\gamma_f > \underline\gamma_f$.
\end{lem}

\begin{thm}
\emph{\textbf{(Globally exponentially stable gyro-bias observer):}}\label{thm1}
Let $\Lambda_{i}=\Lambda^{T}_{i} >0$, $k_{i} >0$, for $i=1,2,\cdots , n$, and $\gamma_{f} >0$ be chosen according to Lemma \ref{lem1} such that $\lambda_o \vcentcolon = \lambda_{\min}(K_{o}) - \epsilon_{f}\sum_{i=1}^{n}k_{i}\lambda_{\max}(\Lambda_i) > 0 $, 
where $K_o := \sum_{i=1}^{n}k_{i} S^{T}(v_{i})\Lambda_i S(v_{i})$
Then, the gyro-bias observer \eqref{eq:bb}-\eqref{eq:be} drives $\hat{b}\to b$ exponentially $\forall\hat{b}(0)\in\mathbb{R}^{3}$. 
\end{thm}

\begin{proof}
It follows from \eqref{eq:vi} and \eqref{eq:KinematicsR} that  $\dot{v}_{i} = S(v_{i})\omega$.
Let the bias-estimation error be defined as $\tilde b:=\hat b-b$.  Its time derivative, by using \eqref{eq:bb}, \eqref{eq:be},  \eqref{eq:vfp} and $K_f$ is 
\begin{equation}\label{eq:bEp}
    \dot{\tilde{b}} = \dot{\bar{b}} - \frac{d}{dt}\sum_{i=1}^{n}k_{i}S^{T}(v_{f_i})\Lambda_{i}v_{i} = K_{f}\left( \hat{\omega} - \omega \right) = -K_{f}\tilde b,  \notag
\end{equation}
which can be expressed by adding $\pm K_o$ as
\begin{equation}\label{eq:bTilDp}
    \dot{\tilde{b}} = -K_{f}\tilde{b}= - K_{o}\tilde{b} - \left( \sum_{i=1}^{n}k_{i} S^{T}(v_{f_i} - v_{i})\Lambda_i S(v_{i}) \right)\tilde{b}.
\end{equation}
Note that $K_o=K_o^T>0$ under Assumption A1 (Lemma 2, \cite{tayebi2013inertial}).    

Consider the Lyapunov function candidate $ V_1 = \frac{1}{2}\|\tilde{b}\|^{2}$. 
Its time evolution along \eqref{eq:bTilDp} is
\begin{align}\label{eq:V1p}
    \dot{V}_1 
    &= - \tilde{b}^{T} K_{o}\tilde{b} - \tilde{b}^{T} \left( \sum_{i=1}^{n}k_{i} S^{T}(v_{f_i} - v_{i})\Lambda_i S(v_{i}) \right)\tilde{b} \notag \\
    &\leq - \lambda_{\min}(K_{o})\|\tilde{b}\|^{2} + \| \sum_{i=1}^{n}k_{i} S^{T}(v_{f_i} - v_{i})\Lambda_i S(v_{i}) \| \|\tilde{b}\|^{2} \notag \\
    &\leq - \lambda_{\min}(K_{o})\|\tilde{b}\|^{2} + \sum_{i=1}^{n}k_{i} \| v_{f_i} - v_{i}\| \|\Lambda_i \| \|\tilde{b}\|^{2} \notag \\
    &\leq - \Big( \lambda_{\min}(K_{o}) + \epsilon_{f}  \sum_{i=1}^{n}k_{i} \lambda_{\max}(\Lambda_i )\Big) \|\tilde{b}\|^{2} \notag \\
    &= -\lambda_{o} \|\tilde{b}\|^{2} = -2\lambda_{o}V_{1}.
\end{align}
Therefore, $\tilde{b}\to 0_{3\times 1}$ exponentially for any initial condition $\tilde{b}(0)\in\mathbb{R}^{3}$ with decaying rate $\lambda_o$.
\end{proof}

\begin{rem}[Gyro-bias observer] 
The gain condition on $\lambda_o$ stated in Theorem \ref{thm1} is different from the high-gain conditions in observer designs (Sec. 14.5, \cite{khalil2002nonlinear}), since $\epsilon_f$ can be arbitrarily small by a sufficiently large filter gain $\gamma_f$ according to Lemma \ref{lem1}. Note that a sufficiently large $\gamma_f$ also implies that the knowledge of the bandwidth of the overall system $\omega_{\max_i}$ is not needed. 
Note that the vector measurement noise may be amplified by the matrix gain $\Lambda_{i}$ due to the second right-side term in \eqref{eq:be} which depends on the vector measurements. Nonetheless, the weights  $k_{i}$ can be assigned to reduce this effect. 
\end{rem}

\section{Attitude Control Design}\label{Sec:CtrlGBEst}
In this section, an attitude tracking controller is first designed based on the bias observer \eqref{eq:bb}-\eqref{eq:be} and assuming known the inertia matrix. The exponential stability, ASGES and AGAS of the desired equilibrium, as well as a separation property of the combined observer and the controller are proved.   Then an adaptive control is devised when the inertia matrix is unknown, and AGAS of the adaptive control is ensured. Relying on a modified bias estimator, only six parameters of the inertia matrix are estimated, and thus no over-parameterization is needed. In both controllers, the vector measurements $v_i$ in \eqref{eq:vi} and the gyro-rate measurement $\omega_g$ with bias estimation $\hat b$   are employed directly in the control law without  attitude estimation.  

\subsection{Control objective and  error variables}
Let $R_{d}\in SO(3)$ be a desired attitude related to a differentiable desired angular velocity $\omega_{d}\in\mathbb{R}^{3}$ through $\dot{R}_{d} = R_{d}S(\omega_{d})$. Define the attitude error as $ R_e = RR^{T}_{d}$, which is associated with its quaternion parameterization $e=\left[e_{0},e^{T}_{v}\right]^{T}\in\mathcal{S}^{3}$ by Rodriguez formula $R_e=R(e)$ in \eqref{eq:Rodrigues}.
 The control objective is to design a control law to achieve asymptotic convergence of the attitude tracking error $e\rightarrow \pm\hat 1=[\pm 1, 0_{1\times 3}]^T$ and the angular velocity error $\omega- \omega_d\rightarrow 0_{3\times 1}$.
 
To relate the quaternion error $e=[e_0, e_v^T]^T$ with the vector measurements, two error variables are defined: the weighted inner product between $v_i$ in \eqref{eq:vi} and its desired value $v_{d_i}:=R^T_d r_i$  
\begin{equation}\label{eq:CfgErr}
    e_{R} = \frac{1}{2}\sum^{n}_{i=1}k_{i}\| v_{i} - v_{d_i}\|^{2} ,
\end{equation} 
where $k_i >0$ are given in the observer \eqref{eq:bb}-\eqref{eq:be}, 
and the weighted cross product between $v_i$ and $v_{d_i}$ 
\begin{equation}\label{eq:z}
z = \sum^{n}_{i=1} k_i S(v_{i} )v_{d_i}. 
\end{equation}

Some useful results for the subsequent analysis are given in the following lemma. Its proof is given in Appendix \ref{App1}.
\begin{lem}\label{lem2} \emph{\textbf{(Properties of error variables $e_R$ and $z$):}} 
\begin{enumerate}
\item[1)] The error $e_R$ in \eqref{eq:CfgErr} can be expressed as
\begin{equation}\label{eq:CfgErr1}
    e_R =  \sum^{n}_{i=1}k_{i}\left( 1-v^{T}_{i}v_{d_i} \right) = 2 e^{T}_{v}We_{v},
\end{equation}
where  $W  \vcentcolon = -\sum^{n}_{i=1} k_i S^{2}(r_{i} )=W^T>0$ under Assumption A1.  
Furthermore,  $0 \leq e_R \leq 2\sum^{n}_{i=1} k_i$. $e_R = 0$ implies $v_{i} = v_{d_i}$ and $e =\pm\hat 1$; while $e_{R} = 2\sum^{n}_{i=1} k_i$ implies $v_{i} = -v_{d_i}$. Besides, $e_{R}=2\lambda_{w,1}$ provided that $e_{v}=v_{w,1}$, where $v_{w,j}$, $j=1,2,3$, are the unit eigenvectors of $W$ associated with the eigenvalues $\lambda_{w,j}$, i.e., $W v_{w,j} = \lambda_{j}v_{w,j}$, ordered as $\lambda_{w,1}\geq\lambda_{w,2}\geq\lambda_{w,3}>0$. 

\item[2)] The time derivative of $e_R$ is
\begin{equation}\label{eq:cfgP}
\dot{e}_{R} = z^{T} (\omega - \omega_d).
\end{equation}

\item[3)] The error variable $z$ in \eqref{eq:z} can be expressed as
\begin{equation}\label{eq:ze}
z = 2R^{T}_{d} \Big( e_{0}I_{3} - S(e_{v})\Big) W e_{v}.
\end{equation}
Furthermore,  $z=0_{3\times 1}$ implies that $e=\pm\hat 1$, or $e=[0, v^T_{w,j}]^T$, $j=1,2,3$. 

\item[4)] The dynamics of $z$ is described by
\begin{equation}\label{eq:zP}
\dot{z} =  J ( \omega - \omega_d ) + S(z)\omega_d ,
\end{equation}
where $J = \sum^{n}_{i=1}k_{i}S^{T}(v_{d_i})S(v_{i})$, with  $\|J\|\leq \sum^{n}_{i=1} k_i $.

\item[5)] For any  $\alpha_1>0$,  there exists $\beta>0$,  such that
\begin{equation}\label{eq:bcond}
      \alpha_{1}e_{R}\leq \frac{\beta}{2}\| z \|^{2}, \ \ \forall t\geq 0 ,
\end{equation}
for all $e \in  \mathcal{S}^{3} \backslash \bigcup B_j$, where $\bigcup B_j \vcentcolon = B_{1}\cup B_{2}\cup B_{3}$, and
{\small
\begin{equation}\label{eq:BallQ}
B_j \vcentcolon = \left\{ e\in\mathcal{S}^{3} \;|\; e=\left[ -\rho , \sqrt{1-\rho^{2}}v^{T}_{w,j} \right]^{T}, \  \rho\in [0,\epsilon]\right\}, \notag
\end{equation}
}
is a closed ball with center in $e_{j}\vcentcolon =\left[ 0,v^{T}_{w,j}\right]^{T}$, and  (arbitrarily small) radius $\epsilon >0$, for $j=1,2,3$.
\end{enumerate}
\end{lem}

\subsection{Attitude tracking controller }
Define the reference velocity $\omega_r: =  -\lambda_{c} z + \omega_d $ and the composite tracking error $\sigma := {\omega} - \omega_r$, where $\lambda_{c} >0$ is a design parameter. Note that the control objective is achieved if $z(t)$ and $\sigma(t) \rightarrow 0_{3\times 1}$ asymptotically. The following control law is proposed
\begin{equation}\label{eq:CtrlOF}
 \tau = M \dot{\hat{\omega}}_r - S( M\hat{\omega} ) \omega_r - K_c {{{{{\hat\sigma}}}}} - \left( \alpha_{1} I_{3} + \alpha_{2}J^{T}  \right) z ,
\end{equation}
where $\dot{\hat{\omega}}_r = -\lambda_{c}J\left( \hat{\omega} - \omega_{d} \right) -\lambda_{c}S(z)\omega_{d} +\dot{\omega}_{d} $, ${{{{{\hat\sigma}}}}} = \hat{\omega} - \omega_r$, $\hat\omega=\omega_g-\hat b$, with  $\hat b$ given by the gyro-bias observer \eqref{eq:be}. 
 $K_c \in \mathbb{R}^{3\times 3}$, $K_c = K^{T}_c >0$ and $\alpha_{1}, \alpha_{2} >0$ are design parameters.

\begin{thm}
\emph{\textbf{(Exponential stability):}}\label{thm2}
Let the design parameters be chosen as follows:
\begin{itemize}
    \item Gyro-bias observer: $\Lambda_{i}=\Lambda^{T}_{i} >0$, $k_{i} >0$, for $i=1,2,\cdots , n$,  and $\gamma_f >0$ according to Theorem \ref{thm1},
     \item Controller:  $K_c = K^{T}_c >0$, $\lambda_{c}>0$ , $\alpha_{1}>0$, $\alpha_{2} >0$,
\end{itemize}
such that
\begin{equation}\label{eq:La}
\lambda_a \vcentcolon = \alpha_{1} -\alpha_{2}\sum^{n}_{i=1}k_{i} > 0 ,
\end{equation}
\begin{equation}\label{eq:CondSty}
\lambda_o \vcentcolon = \lambda_{\min}(K_{o}) - \epsilon_{f}\sum_{i=1}^{n}k_{i}\lambda_{\max}(\Lambda_i)> \frac{\|G\|^{2}}{4\lambda_{\min}(K_{c})},
\end{equation}
where $G\vcentcolon = K_{c} - S(\omega_{r})M + \lambda_{c}MJ$, which is bounded in light of the boundedness of $z$ and $\omega_d$.
Then the attitude controller \eqref{eq:CtrlOF} and the gyro-bias observer \eqref{eq:bb}-\eqref{eq:be}, in closed loop with the system \eqref{eq:KinematicsR}-\eqref{eq:DynamicsB}, have the following properties: for $j=1,2,3$ 
\renewcommand{\labelenumi}{(\roman{enumi})}
\begin{enumerate}
    \item The set of equilibria for $ \big( e,\sigma ,\tilde{b} \big)  \in \mathcal{S}^{3}\times\mathbb{R}^{3}\times\mathbb{R}^{3}$ is $\big\lbrace \big( \pm \hat{1},0_{3\times 1} , 0_{3\times 1} \big) , \big( e_{j} , 0_{3\times 1}, 0_{3\times 1} \big)  \big\rbrace$.
    
    \item The equilibrium $\big( \pm \hat{1},0_{3\times 1}, 0_{3\times 1} \big)$ is exponentially stable  
$\forall \big( e(0),\sigma(0) ,\tilde{b}(0) \big)\in \mathcal{X}:= \mathcal{S}^{3} \backslash B_e \times\mathbb{R}^{3}\times\mathbb{R}^{3}$, where $B_e:=\{ e \in \mathcal S^3| \ e_R<2\lambda_{\omega,3}\}$.
    
    \item The three undesired equilibria $\big( e_{j} , 0_{3\times 1}, 0_{3\times 1} \big)$ are unstable.
\end{enumerate}
\end{thm}
\begin{proof}
The proof starts with calculating the dynamics of the error state $x_1:=[z^T\ \sigma^T \ \tilde b^T]^T$ in closed loop with the control law \eqref{eq:CtrlOF}. By \eqref{eq:zP}, the dynamics of $z$ is rewritten as
\begin{equation}\label{eq:zP2}
\dot{z} = J \sigma -\lambda_c J z + S(z)\omega_d.
\end{equation}
The dynamics of $\sigma$ can be calculated by the definition of $\omega_r$ and its time derivative $\dot \omega_r=-\lambda_c J(\omega-\omega_d)-\lambda_c S(z)\omega_d+\dot \omega_d$, the definition of $\dot{\hat\omega}_r$, the fact $\hat\omega-\omega=-\tilde b$, the property of the skew-symmetric matrix, and the control law \eqref{eq:CtrlOF} as follows
\begin{align}\label{eq:OmTilP}
& M\dot\sigma = M(\dot{\omega} - \dot{\omega}_r) 
= S(M\omega)\omega +\tau -M\dot\omega_r \notag\\
=&S(M\omega)\sigma +\tau -M\dot\omega_r+S(M\omega)\omega_r \notag\\
=&S(M\omega)\sigma - K_c \sigma + G \tilde{b} - \left( \alpha_{1} I_{3} + \alpha_{2}J^{T}  \right) z.
\end{align}
The error dynamics of the closed-loop system is completed by the dynamics of the bias-estimation error in \eqref{eq:bTilDp}.  Note that $x_1=0_{9\times 1}$ is the only equilibrium of the closed loop system \eqref{eq:zP2}, \eqref{eq:OmTilP}, and \eqref{eq:bTilDp}. Therefore, according to item (3) of Lemma \ref{lem2}, the set of equilibria for $ \big( e,\sigma ,\tilde{b} \big)  \in \mathcal{S}^{3}\times\mathbb{R}^{3}\times\mathbb{R}^{3}$ is given in item (i).

To prove item (ii), consider the following Lyapunov function candidate 
\begin{equation}\label{eq:V2}
V_{2}(x_1,t) = \frac{1}{2}\sigma^{T}M\sigma + \frac{1}{2}\| \tilde{b} \|^{2} + \frac{\alpha_2}{2} \| z \|^{2} + \alpha_{1} e_R ,
\end{equation}
$\forall x_1=[ z^T,\sigma^T ,\tilde{b}^T]^T\in   \mathbb{R}^{3} \times\mathbb{R}^{3}\times\mathbb{R}^{3} $. Note that in the region $\mathcal X$ the equilibrium $x_1=0_{9\times 1} \implies e=\pm\hat 1$ because $2\lambda_{w,3}\|e_v\|^2\leq 2e_v^T We_v= e_R<2\lambda_{\omega,3}$ implies $\|e_v\|<1$, and therefore the undesired equilibria are excluded from the region $\mathcal X$. $V_{2}$ is bounded by
\begin{align}
\gamma_{1}\|x_1\|^{2} &\leq V_{2}(x_1,t)\leq \gamma_{2}\|x_1\|^{2} , \label{eq:V2b}
\end{align}
with $ \gamma_{1}=\frac{1}{2}\min \left\{ \lambda_{\min}(M) , 1, \alpha_2  \right\}$, and 
$\gamma_{2}= \frac{1}{2}\max\left\{ \lambda_{\max}(M), 1, \alpha_{2}+\beta \right\}$  for some $\beta >0$ according to item (5) of Lemma \ref{lem2}.

The time derivative of \eqref{eq:V2} along the error dynamics \eqref{eq:zP2}, \eqref{eq:OmTilP}, and \eqref{eq:bTilDp},  
and the time derivative of $e_R$ in \eqref{eq:cfgP} written as $\dot{e}_{R} = z^{T} (\omega - \omega_d) = z^{T} \sigma-\lambda_c z^{T}z $, is 
\begin{align}\label{eq:V2p}
\dot{V}_2 &= \sigma^{T}M \dot\sigma + \tilde{b}^{T}\dot{\tilde{b}} + \alpha_{2}z^{T}\dot{z} + \alpha_{1} \dot{e}_R \notag\\
  &= -\sigma^{T} K_c \sigma   +\sigma^{T}G\tilde{b} - \tilde{b}^{T} K_{f}\tilde{b} -\lambda_{c}z^{T} \left( \alpha_{1} I_{3} + \alpha_{2} J \right) z \nonumber \\
  &= - x_1^{T}Q_{1}x_1 \leq -\lambda_{\min}(Q_{1})\|x_1\|^{2} \leq -\frac{\lambda_{\min}(Q_{1})}{\gamma_{2}}V_{2},
\end{align}
where  
\begin{equation}\label{eq:Q1}
    Q_{1} = \left[
	\begin{array}{ccc}
		\lambda_{c}\left( \alpha_{1} I_{3} + \alpha_{2} J \right)  & 0  & 0  \\
		 & K_c &-\frac{1}{2}G \\
		0 & -\frac{1}{2}G & K_{f} 
	\end{array}	  
  \right] ,
\end{equation}
is positive definite under conditions \eqref{eq:La} and \eqref{eq:CondSty}. Therefore, by Theorem 4.10 of \cite{khalil2002nonlinear}  it follows that
\begin{equation*}
    \|x_1(t)\|\leq \left( \frac{\gamma_{2}}{\gamma_{1}}\right)^{1/2}\|x_1(t_{0})\|\mathrm{e}^{-\left( \frac{\lambda_{\min}(Q_{1}}{2\gamma_{2}})\right)(t-t_{0})}.
\end{equation*}
Then the equilibrium $x_1 =0_{9\times 1}$ is globally exponentially  stable $\forall x_1(0)= [z^T(0),\sigma^T(0) ,\tilde{b}^T(0)]^T\in  \mathbb R^3 \times\mathbb{R}^{3}\times\mathbb{R}^{3} $. This implies that  $\big( e,\sigma,\tilde{b}\big)= \big( \pm \hat{1},0_{3\times 1}, 0_{3\times 1} \big)$ is (locally) exponentially stable  for $\big( e(0),\sigma(0) ,\tilde{b}(0)\big) \in \mathcal{X}$. 

To prove item (iii), that is, the instability of the undesired equilibria $e= e_{j}$, $j=1,2,3$,  the Lyapunov function $V_{2}$ near the equilibrium $e= e_{j}$ is analyzed. 
With little notation abuse, denote $V_2(x_1,t)$ by $V_2(e,\omega_e,\tilde b)$. It is sufficient to show, by Chetaev's Theorem (Theorem 4.3, \cite{khalil2002nonlinear}), that if $V_{2}\left(e^{*}_{j}, \sigma^{*},\tilde{b}^{*}\right) < V_{2}\left( e_{j}, 0_{3\times 1},0_{3\times 1}\right)$, for some $\left(e^{*}_{j},  \sigma^{*},\tilde{b}^{*}\right)$ arbitrarily close to the undesired equilibrium $\left(e_{j}, 0_{3\times 1},0_{3\times 1}\right)$, then the equilibrium $\left(e_{j}, 0_{3\times 1},0_{3\times 1}\right)$ is unstable since $V_{2}$ is a non-increasing function. 

Following \eqref{eq:SmllR} in Appendix \ref{App1}, a small rotation on $e_{j}$ is given by
\begin{equation}\label{eq:SmllR2}
    e^{*}_{j} =   
    \left[ \begin{array}{c}
         0\\
         v_{w,j} 
    \end{array} \right]\otimes \left[ \begin{array}{c}
         \sqrt{1-\epsilon^{2}}  \\
        \epsilon v_{w,j}
    \end{array} \right]^{T} =
    \left[ \begin{array}{c}
         -\epsilon \\
        \left( \sqrt{1-\epsilon^{2}}\right) v_{w,j}
    \end{array} \right],
\end{equation}
for an arbitrarily small  $1>\epsilon>0$, where $\otimes$ denotes the quaternion product. 

Define $\Delta V \vcentcolon = V_{2}\big( e^{*}_{j}, \sigma^{*},\tilde{b}^{*}\big) - V_{2}\big(e_{j}, 0_{3\times 1},0_{3\times 1}\big)$. By \eqref{eq:V2}, \eqref{eq:SmllR2}, and \eqref{eq:NormZ}-\eqref{eq:NormEr} in Appendix \ref{App1}, it has
\begin{align*}
 \Delta V_2 &=
    \frac{1}{2}\sigma^{*T}M\sigma^{*} + \frac{1}{2} \tilde{b}^{*T}\tilde{b}^{*} + 2\alpha_{2}\epsilon^{2}\lambda^{2}_{j}\left( 1-\epsilon^{2} \right) \notag \\
    &\quad + 2 \alpha_{1} \left( 1-\epsilon^{2} \right)\lambda_{j} - 2 \alpha_{1}\lambda_{j} \notag  \\
    &= \frac{1}{2}\sigma^{*T}M\sigma^{*} + \frac{1}{2} \tilde{b}^{*T}\tilde{b}^{*} + 2\alpha_{2}\epsilon^{2}\lambda^{2}_{j}\left( 1-\epsilon^{2} \right) -2\alpha_{1}\lambda_{j}\epsilon^{2} \notag \\
    &= \frac{1}{2}\sigma^{*T}M\sigma^{*} + \frac{1}{2} \tilde{b}^{*T}\tilde{b}^{*} - 2 \lambda_{j} \left(  \alpha_{1} - \alpha_{2}\lambda_{j}\left( 1-\epsilon^{2} \right)\right)\epsilon^{2} \\
    &<0 ,
\end{align*}
 for $ \frac{1}{4 \lambda_{j} \left(  \alpha_{1} - \alpha_{2}\lambda_{j}\left( 1-\epsilon^{2} \right)\right)} \left(  \sigma^{*T}M\sigma^{*} + \tilde{b}^{*T}\tilde{b}^{*}\right)<  \epsilon^{2}$.
Therefore, under the condition  \eqref{eq:La} and  $\forall \epsilon>0$, there always exists
$\left(e^{*}_{j},  \sigma^{*},\tilde{b}^{*}\right)$ arbitrarily close to the undesired equilibrium $\left(e_j, 0_{3\times 1},0_{3\times 1}\right)$, such that $V_{2}\left(e^{*}_{j}, \sigma^{*},\tilde{b}^{*}\right) < V_{2}\left(e_j, 0_{3\times 1},0_{3\times 1}\right)$. This shows the instability of the three undesired equilibria.
\end{proof}

Note that the attraction region in item (ii) of Theorem \ref{thm2} may be enlarged to cover almost the entire space of $ \mathcal{S}^{3} \times\mathbb{R}^{3}\times\mathbb{R}^{3}$, as the following corollary states. 

\begin{defn}[Almost global asymptotic stability  (AGAS) and almost semi-global exponential stability (ASGES)) \cite{lee2015global}]
 Let the origin be an equilibrium of a dynamic system. The equilibrium is of:
\begin{itemize}
    \item[i) ] AGAS, if it is asymptotically stable and the set of initial states that do not converge to the origin has zero Lebesgue measure.
\item[ii) ] ASGES, if it is asymptotically stable and for almost all initial states, there exist finite controller parameters such that the corresponding trajectory exponentially converges to the origin.
\end{itemize}
\end{defn}

\begin{col}
\emph{\textbf{(ASGES and AGAS):}}\label{prop1}
The equilibrium  $\big( \pm \hat{1},0_{3\times 1}, 0_{3\times 1} \big)$ is of AGAS $\forall \big( e(0),\sigma(0) ,\tilde{b}(0) \big)\in  \mathcal{S}^{3}\backslash \{e_{j}\} \times\mathbb{R}^{3}\times\mathbb{R}^{3} $, and ASGES 
    $\forall \big( e(0),\sigma(0) ,\tilde{b}(0) \big)\in  \mathcal{S}^{3} \backslash \bigcup B_j \times\mathbb{R}^{3}\times\mathbb{R}^{3} $,   where  $e_{j}$ and $B_j$ are defined in item 5) of Lemma \ref{lem1}.
\end{col}
\begin{proof}
For $e\in\mathcal S^3$,  the inequality \eqref{eq:bcond} is no longer held, nor the upper bound on the Lyapunov function $V_2$ in \eqref{eq:V2}. It follows from the proof of Theorem \ref{thm2} that $\dot V_2<0$, which shows the asymptotic stability of $x_1=0_{9\times 1}$. Since the Lebesgue measure of the unstable set $\{(e_{j},0_{3\times 1}, 0_{3\times 1})\}$ is zero, the equilibrium  $\big( \pm \hat{1},0_{3\times 1}, 0_{3\times 1} \big)$ is AGAS. On the other hand, $\forall \epsilon>0$, the proof of item (ii) in Theorem \ref{thm2} holds for $e\in \mathcal{S}^{3} \backslash \bigcup B_j$. Note that the set $\mathcal{S}^{3} \backslash \bigcup B_j$ can be arbitrarily enlarged for a given $\alpha_1$ and any arbitrarily small $\epsilon>0$  by a sufficiently large $\beta$ in item (5) of Lemma \ref{lem1} to cover almost entire  $\mathcal S^3$, the equilibrium  $\big( \pm \hat{1},0_{3\times 1}, 0_{3\times 1} \big)$ is thus ASGES.
\end{proof}

\begin{rem}[Comparison with some reported controllers]
By exploiting the relationship between the error variables $z$ and $e_R$ near the  equilibria characterized by item (i) of Theorem \ref{thm2} as established in item 5) of Lemma \ref{lem1}, it allows one to propose a strict Lyapunov function \eqref{eq:V2} for the desired equilibrium. This in turn enables the controller \eqref{eq:CtrlOF} to achieve AGAS and ASGES under the conditions stated in Theorem \ref{thm2}, contrasting with only AGAS in \cite{pounds2007attitude,mercker2011rigid,tayebi2013inertial,forbes2013passivity}. 
\end{rem}

\begin{rem}[A separation property] 
Note that the bias observer \eqref{eq:bb}-\eqref{eq:be} is used in the control loop without modifications. This implies that the bias observer \eqref{eq:bb}-\eqref{eq:be} and controller \eqref{eq:CtrlOF} may be designed separately provided that the convergence rate $\lambda_o$ in \eqref{eq:CondSty} of the observer is sufficiently large. This condition may also be obtained through finite-time convergence, e.g., \cite{hamrah2021finite}.


\end{rem}

\begin{rem}[Controller tuning] 
The controller and the observer can be tuned as follows to satisfy the conditions in Theorem \ref{thm2}: 1) assign the weights $k_i$ to the observer \eqref{eq:bb}-\eqref{eq:be} according to the confidence level of each sensor, 2) choose the controller gains $K_c = K^{T}_c >0$, $\lambda_{c}>0$ , $\alpha_{1}>0$, $\alpha_{2} >0$ such that the condition \eqref{eq:La} holds, 3) choose the filter gain $\gamma_f$ as in Lemmas \ref{lem1} and the matrix gain $\Lambda_i$ in \eqref{eq:bb}-\eqref{eq:be} such that the condition \eqref{eq:CondSty} holds.

\end{rem}

\subsection{Adaptive Attitude Tracking Controller}\label{Sec:CtrlAdp}

To deal with the unknown inertia matrix, the first two terms in the control law \eqref{eq:CtrlOF} are replaced by an adaptive compensation  in the following proposed adaptive control law:
\begin{equation}
\tau = Y(\hat{\omega} ,h)\hat{\theta} -K_{c}\hat{\sigma} - \left( \alpha_{1} I_{3} + \alpha_{2}J^{T}  \right) z ,\label{eq:CtrlAdp}\\
\end{equation}
where $\hat{\sigma} = \hat{\omega}-\omega_{r}$, and the controller gains $K_c,\ \alpha_1, \ \alpha_2$ are the same as in the control law \eqref{eq:CtrlOF}, $Y(\hat{\omega} ,h)$ is the regressor
\begin{align}
Y&(\hat{\omega}, h) := S(\hat{\omega})F_{1}(\hat{\omega})+F_{1}(h) , \label{eq:Yreg} \\
h& := -\lambda_{c}\Big( J\big( \hat{\omega} -\omega_{d} \big) + S(z)\omega_{d} \Big) + \dot{\omega}_{d} + \big( \alpha_{1} I_{3} + \alpha_{2}J^{T}  \big) z , \notag 
\end{align}
$F_{1}(u)$,  for a vector $u= \left[ u_{1},u_{2},u_{3}\right]^T$, is defined as 
\begin{align*}
F_{1}(u) &\vcentcolon = \left[ \begin{array}{cccccc}
u_{1} & 0 & 0 & 0 & u_{3} & u_{2}  \\
0 & u_{2} & 0 & u_{3} & 0 & u_{1}\\
0 & 0 & u_{3} & u_{2} & u_{1} & 0 
\end{array} \right],
\end{align*}
and $\hat{\theta}(t)$ is the parameter estimation of the inertia matrix entries $\theta = \left[ m_{11},m_{22},m_{33},m_{23},m_{13},m_{12} \right]^{T}\in \mathbb{R}^{6}$, 
which is updated according to 
\begin{align}\label{eq:thtE}
    \dot{\hat{\theta}} =&  -\Gamma Y^{T}(\hat{\omega},h)\hat\sigma,
\end{align}
where $\Gamma\in\mathbb{R}^{6\times 6}$, $\Gamma=\Gamma^{T}>0$ is the adaptation gain.  The estimated gyro bias $\hat{b}$ used to calculate $\hat\omega=\omega_g-\hat b$ is obtained by the following modified observer
\begin{equation}\label{eq:beAd}
    \hat{b} = \mu_{b}\mathrm{Tanh}\left(\bar{b}\right) - \sum_{i=1}^{n}k_{i}S^{T}(v_{f_i})\Lambda_{i}v_{i} ,
\end{equation}
{\footnotesize
\begin{equation}\label{eq:bbAd}
\dot{\bar{b}} = \frac{1}{\mu_{b}}\mathrm{Cosh}^{2}\left(\bar{b}\right)\left(K_{f}\hat{\omega}+\sum_{i=1}^{n}k_{i}S(\Lambda_{i}v_{i})\dot{v}_{f_i}  - \left( \alpha_{1} I_{3} + \alpha_{2}J^{T} \right) z \right),  
\end{equation}}
where  $\mu_{b} \geq \theta_{b}$ with $\theta_b$ the bound on $\|b\|$ in Assumption A2. The  hyperbolic functions  are defined entry-wise:
\begin{align*}
\mathrm{Tanh}(\bar{b})&\vcentcolon =[\tanh{(\bar{b}_{1})},\tanh{(\bar{b}_{2})}, \tanh{(\bar{b}_{3}))} ]^{T}\in \mathbb{R}^3 ,\\
\mathrm{Cosh}(\bar{b})&\vcentcolon = \mathrm{diag}\lbrace \cosh{(\bar{b}_{1})},\cosh{(\bar{b}_{2})},\cosh{(\bar{b}_{3})} \rbrace \in \mathbb{R}^{3\times 3}, \\
\mathrm{Sech}(\bar{b})&\vcentcolon = \mathrm{diag}\lbrace \mathrm{sech}{(\bar{b}_{1})},\mathrm{sech}{(\bar{b}_{2})},\mathrm{sech}{(\bar{b}_{3})} \rbrace \in \mathbb{R}^{3\times 3}.
\end{align*}
Note that $\hat b$ in \eqref{eq:beAd} is bounded by $\|\hat b\|\leq \sqrt{3}\mu_b+\sum_{i=1}^n k_i\lambda_{max}(\Lambda_i)$. 
The stability of the proposed adaptive controller is summarized in the following theorem.

\begin{thm}
\emph{\textbf{(Adaptive controller):}}\label{thm3}
Let the design parameters be chosen as follows: 
\begin{itemize}
    \item Gyro-bias observer: $\Lambda_{i}=\Lambda^{T}_{i} >0$, $k_{i} >0$, for $i=1,2,\cdots , n$, and $\gamma_f >0$, $\mu_{b} \geq \theta_{b}$,
     \item Adaptive controller:  $K_c = K^{T}_c >0$, $\lambda_{c}>0$ , $\alpha_{1}>0$, $\alpha_{2} >0$, $\Gamma=\Gamma^{T}>0$,
\end{itemize}
such that condition \eqref{eq:La} is satisfied and
\begin{align}
&\lambda_o  = \lambda_{\min}(K_{o}) - \epsilon_{f}\sum_{i=1}^{n}k_{i}\lambda_{\max}(\Lambda_i) > 0, \label{eq:CondAdp}\\      
&\lambda_{\min}(K_{c})  >  \frac{\| H \|^{2}}{4\lambda_{o}}+\|S(\tilde{b})M\| , \label{eq:CondAdp1}
\end{align}
where $ H \vcentcolon = M\left( K_{f} + \lambda_{c}J \right) +S(M\omega_{r}) -S(\tilde{b}+\omega_{r})M$, which is bounded.
 Then the adaptive controller \eqref{eq:CtrlAdp}, the gyro-bias observer \eqref{eq:beAd},  and the adaptation law \eqref{eq:thtE}, in closed loop with the system \eqref{eq:KinematicsR}-\eqref{eq:DynamicsB}, have the following properties, for $j=1,2,3$: 
 \renewcommand{\labelenumi}{(\roman{enumi})}
 \begin{enumerate}
    \item All variables are bounded, and $\left( e,\hat \sigma ,\tilde{b} \right) \rightarrow \left( \pm \hat{1},0_{3\times 1} , 0_{3\times 1} \right) $  almost globally asymptotically 
    $\forall\left( e(0),\hat\sigma(0) ,\tilde{b}(0) , \tilde{\theta}(0) \right)\in \mathcal{X}_{a}\vcentcolon =  \mathcal{S}^{3} \backslash \bigcup B_j \times\mathbb{R}^{3}\times\mathbb{R}^{3}\times\mathbb{R}^{6}$,  where $\tilde\theta=\hat\theta-\theta$ is the parameter estimation error.
    \item The three undesired equilibria $\big( e_{j} , 0_{3\times 1}, 0_{3\times 1},0_{6\times 1} \big)$ are unstable. 
\end{enumerate}
\end{thm}
\begin{proof} Let  $x_2:=[z^T\ {\hat\sigma}^T \ \tilde b^T \ \tilde\theta^T]^T$ denote the state of the closed loop system with the control law \eqref{eq:CtrlAdp}. Its dynamics is calculated as follows. 

The dynamics of the bias estimation error by \eqref{eq:beAd} and \eqref{eq:bbAd} is 
\begin{align}\label{eq:bTilDp_Ad}
    \dot{\tilde{b}} &= \dot{\hat{b}} - \dot{b} = \mu_{b}\mathrm{Sech}^{2}\left(\bar{b}\right)\dot{\bar{b}} -K_{f}\omega - \sum_{i=1}^{n}k_{i}S(\Lambda_{i}v_{i})\dot{v}_{f_i} \notag\\
    &=-K_{f}\tilde{b} - \left( \alpha_{1} I_{3} + \alpha_{2}J^{T}  \right) z  .
\end{align}

The dynamics of $\hat\sigma=\hat\omega-\omega_r$,  in view of \eqref{eq:VelM} and \eqref{eq:bTilDp_Ad}, is now given by 
\begin{align*}
M\dot{\hat\sigma} &= M\dot{\hat{\omega}} - M\dot{\omega}_{r} = M\dot{\omega} - M\dot{\omega}_{r} - M\dot{\hat{b}} \\
    =& S(M\omega)\omega +\tau -Mh + M\left( K_{f} + \lambda_{c}J\right) \tilde{b} \nonumber \\
    =& -\Big( S(\omega)F_{1}(\omega)+F_{1}(h)\Big) \Theta   + M\left( K_{f} + \lambda_{c}J\right)  \tilde{b} +\tau\\
     =& -Y(\omega, h)\Theta +  M\left( K_{f} + \lambda_{c}J\right)  \tilde{b} + \tau.
\end{align*}
Adding and subtracting $Y(\hat{\omega},h)\theta$ gives
\begin{align*}
    M\dot{\hat\sigma}&=  -Y(\omega, h)\theta +Y(\hat{\omega},h)\theta - Y(\hat{\omega},h)\theta \\
    &\quad + M\left( K_{f} + \lambda_{c}J\right)  \tilde{b} + \tau \\
    &= S(\hat{\omega})M\hat{\omega} - S(\omega)M\omega - Y(\hat{\omega},h)\theta \\
    &\quad + M\left( K_{f} + \lambda_{c}J\right)  \tilde{b} + \tau \\
    &= \left( S(M\hat{\sigma})-S(\hat{\sigma})M +S(M\omega_{r}) - S(\tilde{b}+\omega_{r})M \right)\tilde{b} \\
    &\quad - Y(\hat{\omega},h)\theta  + M\left( K_{f} + \lambda_{c}J\right)  \tilde{b} + \tau \notag\\
    &= -S(\tilde{b})M\hat{\sigma} -S(\hat{\sigma})M\tilde{b} - Y(\hat{\omega},h)\theta  + H\tilde{b} + \tau ,
\end{align*}
which in closed loop with the controller \eqref{eq:CtrlAdp} gives 
\begin{align}\label{eq:OmSp}
     M\dot{{{{{\hat\sigma}}}}} &= Y(\hat{\omega}, h) \tilde{\theta} - \left( K_{c} + S(\tilde{b})M\right)\hat{\sigma} -S(\hat{\sigma})M\tilde{b} \notag\\
     &\quad - \left( \alpha_{1} I_{3} + \alpha_{2}J^{T}  \right) z  + H \tilde{b}.
\end{align}

The error state dynamics is completed with the following
\begin{align}
\dot{z} &= J \hat{\sigma} -\lambda_c J z + S(z)\omega_d +J\tilde{b}, \label{eq:zP3}\\
\dot {\tilde \theta}& = -\Gamma Y^{T}(\hat{\omega},h)\hat\sigma. \label{eq:tildetheta}
\end{align}


Note that the state $x_2$ of the closed-loop system \eqref{eq:zP3},  \eqref{eq:OmSp}, \eqref{eq:bTilDp_Ad}, and \eqref{eq:tildetheta}  has the only equilibrium at $x_2=0_{15\times 1}$. To determine its stability, consider the following Lyapunov function candidate 
\begin{equation}\label{eq:V3}
    V_{3} = \frac{1}{2}\hat\sigma^{T}M\hat\sigma + \frac{1}{2}\| \tilde{b} \|^{2} + \frac{1}{2}  \tilde{\theta}^{T}\Gamma^{-1}\tilde{\theta} + \frac{\alpha_2}{2} \| z \|^{2} + \alpha_{1} e_R ,
\end{equation}
$\forall x_2 \in  \mathbb{R}^{3}  \times\mathbb{R}^{3}\times\mathbb{R}^{3} \times \mathbb{R}^{6} $ and $e\in \mathcal{S}^{3} \backslash \bigcup B_j$.

The time evolution of \eqref{eq:V3} along the closed-loop dynamics, by using the fact that $\dot{e}_{R} = z^{T} \hat\sigma-\lambda_c z^{T}z +z^{T}\tilde{b}$, is given by
\begin{align}\label{eq:V3p}
    \dot{V}_{3} &= \hat\sigma^{T}M \dot{\hat\sigma} + \tilde{b}^{T}\dot{\tilde{b}} + \tilde{\theta}^{T}\Gamma^{-1}\dot{\tilde{\theta}}   +  \alpha_{2}z^{T}\dot{z} + \alpha_{1} \dot{e}_R \notag\\
    &= - \hat\sigma^{T}( K_{c} + S(\tilde{b})M)\hat\sigma + \hat\sigma^{T}H \tilde{b} -\tilde{b}^{T}K_{f}\tilde{b} \notag\\
    &\quad -\lambda_c z^{T}\left( \alpha_{1}I_{3} + \alpha_{2} J \right) z \notag\\
      &\leq -\bar{x}^{T} Q_{2}\bar{x},
\end{align}
where $\bar{x} \vcentcolon = \left[ \|z\|, \|\hat\sigma \|,  \| \tilde{b} \| \right]^{T}$, and
\begin{equation}\label{eq:Q2}
    Q_{2} = \left[
    \begin{array}{ccc}
         \lambda_{c}\lambda_{a}  &0  & 0  \\
         0& \lambda_{\min}(K_{c})-\|S(\tilde{b})M\| &  -\frac{1}{2}\|H\| \\
         0 & -\frac{1}{2}\|H\|  &  \lambda_{o} 
    \end{array}
    \right] ,
\end{equation}
which is semipositive definite under conditions \eqref{eq:CondAdp1}-\eqref{eq:CondAdp}, and $\dot{V}_{3} \leq 0$. This implies that the state $x_2$ bounded. Next, it is straightforward to verify that $\ddot{V}_3$ is bounded, and then $\left( z,\hat\sigma ,\tilde{b}\right) \rightarrow \left( 0_{3\times 1},0_{3\times 1} , 0_{3\times 1} \right)$ asymptotically by invoking Barbalat's Lemma. Moreover, $z\rightarrow0_{3\times 1}$ implies that $\left( e,\hat\sigma ,\tilde{b} \right) \rightarrow \left( \pm \hat{1},0_{3\times 1} , 0_{3\times 1} \right)$  almost globally asymptotically, and $ \tilde\theta$ remains bounded. Furthermore, the three undesired equilibria $(e_{j},0_{3\times 1} , 0_{3\times 1},0_{6\times 1})$ are unstable by the same arguments as in the proof of Theorem \ref{thm2}.  
\end{proof}

\begin{rem}[The Adaptive Controller] 
Compared to some similar adaptive control laws for attitude tracking of rigid bodies, the proposed scheme does not require an attitude estimation procedure, a property leveraged from the nonadaptive design in the previous subsection. Moreover,  in the proposed adaptive controller only 6 parameters of the unknown inertia matrix are estimated compared with the over-parameterization required in  \cite{mercker2011rigid,benallegue2018adaptive}. 
Furthermore, in the adaptation law, neither restrictions on the smoothness of the desired trajectory nor parameter projections are needed as opposed to those used in \cite{benallegue2018adaptive}. 
Note that when a persistent excitation condition on $Y\left( \hat{\omega},h\right)$, the exponential convergence of  $\tilde{\theta}\to 0_{6\times 1}$ can be established as in classical adaptive control designs  \cite{ioannou2012robust}. 
\end{rem}

\begin{rem}[Tuning procedure] 
The design parameters may be tuned to verify conditions \eqref{eq:La} and \eqref{eq:CondAdp}-\eqref{eq:CondAdp1} as follows: 1) Assign the weight $k_{i}$ according to the confidence  of each vector measurement $v_i$. 2) Choose the observer gains $\Lambda_{i}=\Lambda^{T}_{i} >0$, $\gamma_{f} >0$ to verify condition \eqref{eq:CondAdp}. Likewise, the observer bound $\mu_{b} >0$ is chosen to be $\mu_{b} \geq \theta_{b}$ according to  Assumption  A2. 3) The controller gains  $\alpha_{1}>0$ and  $\alpha_{2}>0$ are selected to fulfill the condition \eqref{eq:La}. 4) Fix a controller gain $\lambda_{c}>0$, and choose $K_{c}=K^{T}_{c}>0$ to verify \eqref{eq:CondAdp1}. 5) The design parameter $\Gamma=\Gamma^{T}>0$ in the adaptive law may be set independently of the design parameters of the observer and the controller.
\end{rem}

\section{Simulations}\label{Sec:Sim}

The adaptive controller \eqref{eq:CtrlAdp} under noisy measurements was simulated. Three inertial reference vectors  $r_{1} = \left[ 0,0,1 \right]^{T}$, $r_{2} = \left( 1/\sqrt{3}\right) \left[ 1,1,1 \right]^{T}$ and $r_{3} = r_{1}\times r_{2}/\|r_{1}\times r_{2}\|$ were considered. The desired trajectory was chosen as in \cite{mercker2011rigid} given by $\omega_{d} = [ \cos(t)+0.5\cos(0.2t) , \; 0.75\sin(2t) ,\; \sin\left( 5t\mathrm{e}^{-0.001t} \right) + \cos(0.5t) ]^{T}$ (rad/s). The initial conditions were $q(0)= [-1,0,0,0]^{T}$, $\omega (0)=[0,0,0]^{T}$ (rad/s), $q_{d}(0)=[0.8,0,0.6,0]^{T}$, and $\omega_{d}(0)=[1.5,0,1]^{T} $ (rad/s).

The parameters are as follows. The inertia parameter vector was $\theta = \left[ 0.0360,0.0869,0.0935, 0.0004,0.0015,-0,0007\right]^{T}$ (Kgm$^2$).  The vector measurement weights were $k_i=0.1$, and the observer gains were $\Lambda_i=10I_{3}$ for $i=1,2,3$, $\gamma_f=1000$, $\mu_b=1$. The adaptive controller gains were  $K_c=3I_{3}$,  $\lambda_c=1$, 	$\alpha_1=0.1$, $\alpha_2=0.01$, $\Gamma=I_{6}$.

The gyro bias was $b=[0.2,0.1,-0.1]^{T}$ (rad/s). Noise was introduced in the vector measurements and gyro measurements as follows. Vector measurements:  $v_{m,i}=(v_{i}+m_{v}\bar{\nu}_{i})/\| v+m_{v}\bar{\nu}_{i} \|$, $\bar{\nu}_{i}=\nu_{i}/\|\nu_{i}\|$, where $\nu_{i}\in\mathbb{R}^{3}$ are zero-mean Gaussian distributions with unitary variance, for $i=1,2,3$, and $m_{v}\in [0,0.1]$ is a uniform distribution; Gyro rate measurement: $\omega_{m} = \omega + m_{w}\nu_{w} + b$, with  $\nu_{w}\in\mathbb{R}^{3}$ a zero-mean Gaussian distributions and unit-variance, and $m_{w}\in [0,0.1]$ a uniform distribution.

The tracking error $e=q\otimes q_d^{-1}$ is shown in Fig. \ref{fig:NysCase}(a), where the convergence of $e_{0}$ to the set $[-0.99,-1]$ is observed. The norm of the error variable $\|z\|$ is displayed in Fig. \ref{fig:NysCase}(b), which remains in the set $[0,0.02]$ after $20$ (s). Besides, the bias estimation error $\|\tilde b\|$ and the estimated composite error $\|\hat\sigma\|$ are depicted in Fig. \ref{fig:NysCase}(c)-(d). Note that they converge to the set $[0,0.2]$, corresponding to the sum of the maximum noise levels in the vector and gyro measurements. Fig. \ref{fig:NysCase}(e) draws the norm of the estimation error $\|\tilde{\theta}\|$, which converges to the set $[0,0.02]$ under the excitation provided by the desired trajectory. The control   effort is shown in  Fig. \ref{fig:NysCase}(f), which is maintained in $[0,1]$ (Nm). 

\begin{figure}[!t]
	\centering
		\includegraphics[trim = 17mm 0mm 17mm 0mm,clip,scale=0.3]{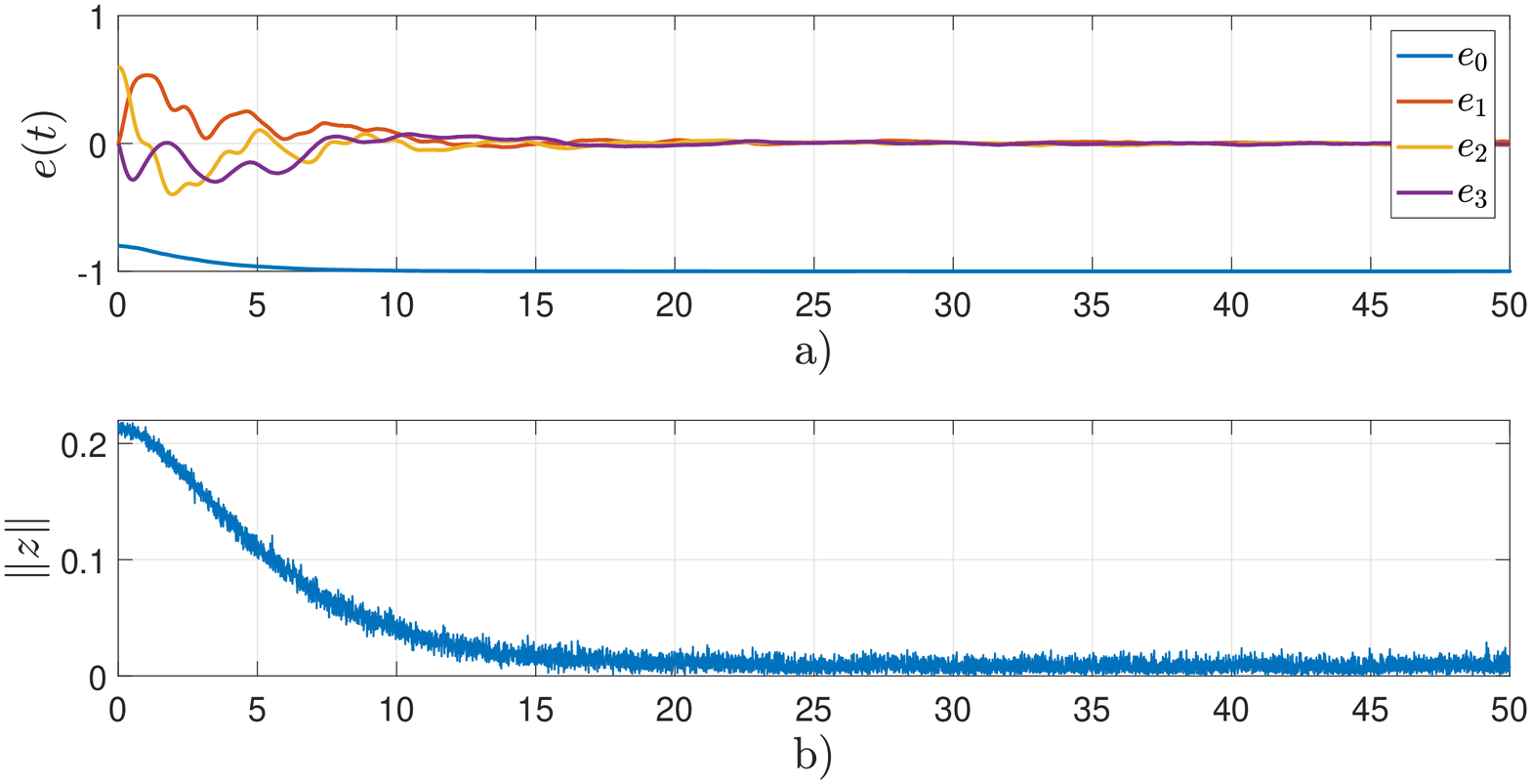}
		\includegraphics[trim = 17mm 0mm 17mm 0mm,clip,scale=0.3]{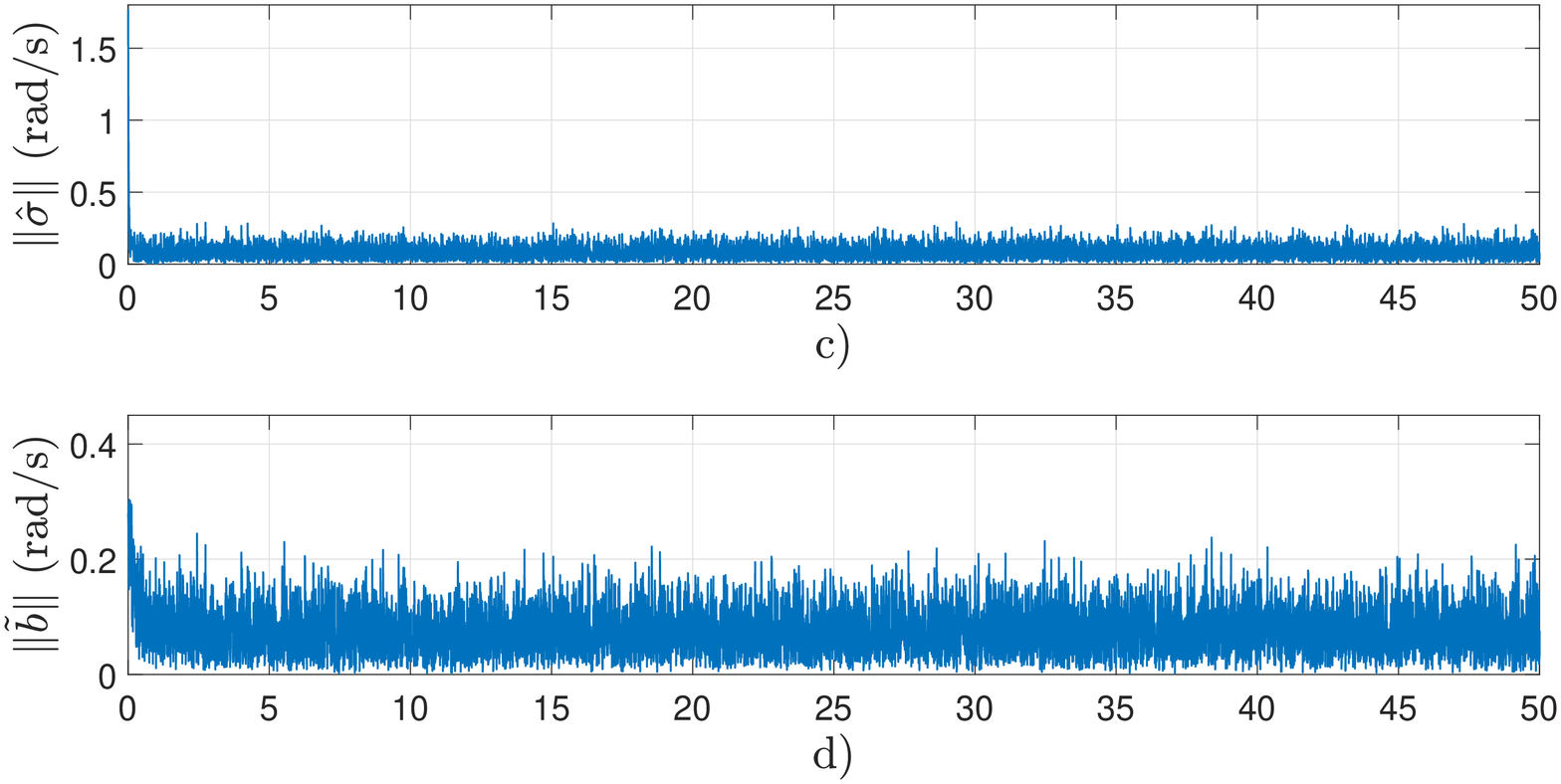}
		\includegraphics[trim = 17mm 0mm 17mm 0mm,clip,scale=0.3]{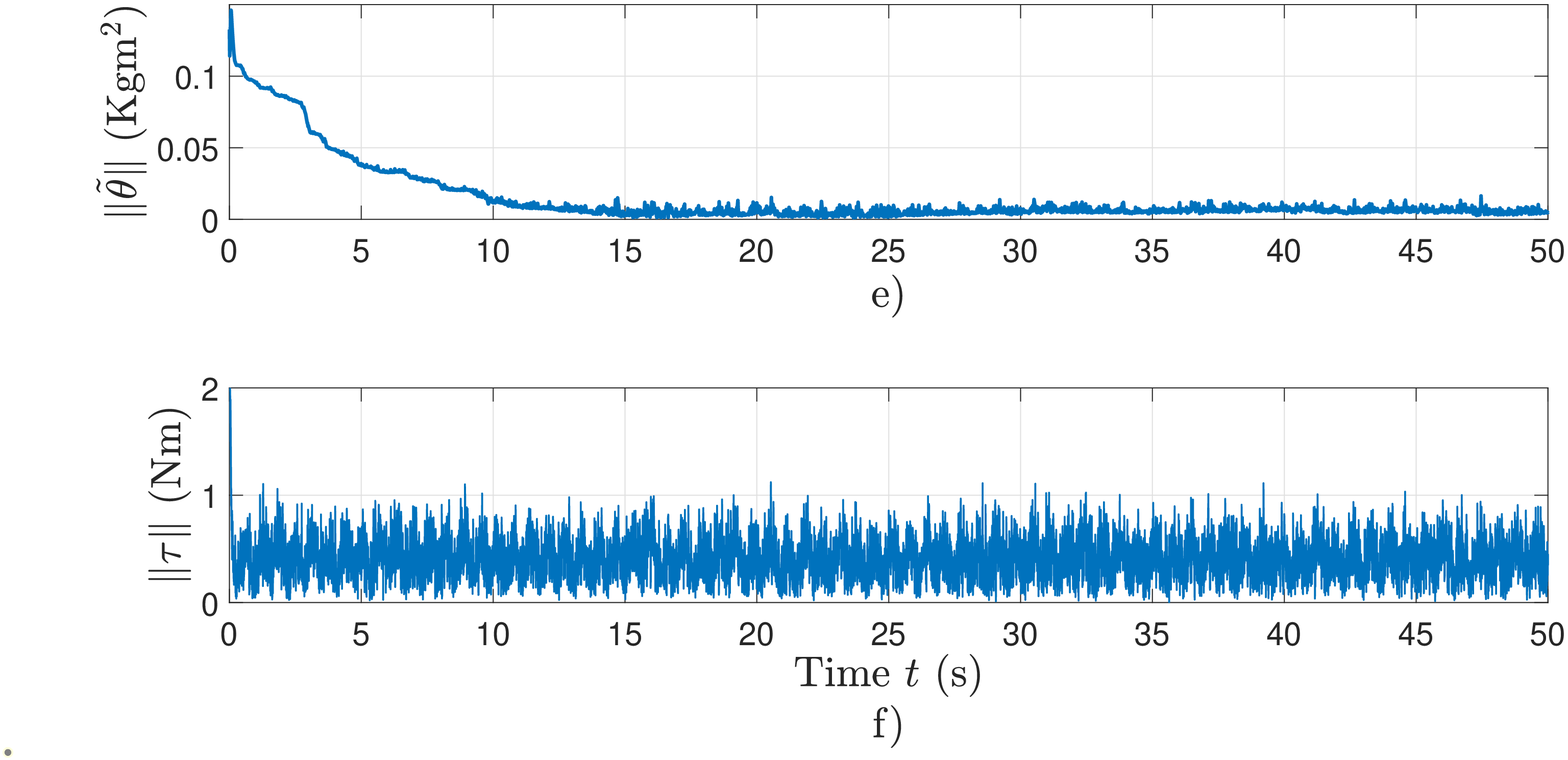}
		\caption{Performance of the adaptive controller  \eqref{eq:CtrlAdp} in the presence of noisy measurements.}
		\label{fig:NysCase}
\end{figure}

\section{Conclusion}\label{Sec:Conc}

In this note, attitude tracking using vector and biased gyro measurements was considered. First, a globally exponentially convergent bias observer was designed. Then, an attitude tracking controller based on the bias observer was developed which is almost globally asymptotically stable and almost semiglobal exponentially stable. A separation property of the combined observer-controller was proved. Lastly, an adaptive controller was devised to deal with the unknown inertia matrix, which estimates only the inertia parameters without over-parameterization thanks to a modified bias observer. Almost global asymptotic stability was demonstrated. Numerical simulations were included to illustrate the theoretical results and verify the robustness of the adaptive controller in the presence of measurement noise. In addition, by directly using the vector measurement, the proposed attitude controllers do not require attitude estimation, which avoids both the undesired critical points introduced by an attitude observer and the unwinding phenomenon in quaternion-based designs. 


%

\appendices
\section{Proof of  Lemma \ref{lem2}} \label{App1}

\begin{enumerate}
\item[1)] Recall the attitude error $ R_e = RR^{T}_{d}$, which is associated with its quaternion parameterization $e=\left[e_{0},e^{T}_{v}\right]^{T}\in\mathcal{S}^{3}$ by Rodriguez formula $R_e=R(e)$ in \eqref{eq:Rodrigues}.  Then it follows from \eqref{eq:CfgErr1}, $v_i=R^Tr_i$, and $v_{d_i}=R_d^Tr_i$ 
that 
\begin{align*}
    e_{R} 
    &= \sum^{n}_{i=1}k_{i}\left( 1 - v^{T}_{i}v_{d_i}\right)= \sum^{n}_{i=1}k_{i}\left( 1 - r^{T}_{i}R_{e}r_{i}\right) \\
    &= \sum^{n}_{i=1}k_{i}\big( 1 - r^{T}_{i} \left( I_{3} + 2e_{0}S(e_{v}) + 2S^{2}(e_{v})\right) r_{i}\Big) \\
    \quad &= -\sum^{n}_{i=1}k_{i}\left(r^{T}_{i} \left( 2S^{2}(e_{v})\right) r_{i}\right) \\
    &= 2e^{T}_{v}\left( -\sum^{n}_{i=1}k_{i}S^{2}(r_{i})\right) e_{v} = 2e^{T}_{v}W e_{v}.
\end{align*}    
Clearly, $0\leq e_{R} = \sum^{n}_{i=1}k_{i}\left( 1 - v^{T}_{i}v_{d_i}\right)\leq 2\sum^{n}_{i=1} k_i$. Furthermore,  $e_R=0\iff v_i=v_{d_i}\iff e=[\pm 1, 0_{1\times 3}]^T$;  $e_{R}=2\sum_{i=1}^n k_i \iff v_{i} = -v_{d_i}$. In particular, $e_{R}=2e^{T}_{v}W e_{v}=2\lambda_{w,1} \iff e_{v}=v_{w,1},$ where $v_{w,j}$, $j=1,2,3$, are the unit eigenvectors of $W$ associated with the eigenvalues $\lambda_{w,j}$, i.e., $W v_{w,j} = \lambda_{j}v_{w,j}$, ordered as $\lambda_{w,1}\geq\lambda_{w,2}\geq\lambda_{w,3}>0$. 


\item[2)] By $\dot v_i=S(v_i)\omega$ and $\dot v_{d_i}=S(v_{d_i}) \omega_d$, the time derivative of $e_R=\sum^{n}_{i=1}k_{i}\left( 1 - v^{T}_{i}v_{d_i}\right)$  is 
\begin{align*}
    \dot{e}_{R} 
    &= -\sum^{n}_{i=1}k_{i}\left( v^{T}_{i}S(v_{d_i})\omega_{d} + v^{T}_{d_i}S(v_{i})\omega \right) \\
    &= -\sum^{n}_{i=1}k_{i}\left(  v^{T}_{d_i}S(v_{i}) (\omega - \omega_{d}) \right) \\
    &= z^{T}(\omega - \omega_{d}).
\end{align*}

\item[3)]The proof is given by Lemma 1 and Lemma 3 of \cite{tayebi2013inertial}.
\item[4)] By the property of the skew-symmetric matrix, the time derivative of the error variable $z$ in \eqref{eq:z} is 
\begin{align*}
    \dot{z} &= \sum^{n}_{i=1}k_{i} \left( S(v_i) \dot{v}_{d_i} - S(v_{d_i})\dot{v}_i \right) \\
    &= \sum^{n}_{i=1}k_{i} \left( S(v_i) S(v_{d_i})\omega_{d} - S(v_{d_i})S(v_{i})\omega \right) \\
    &= \sum^{n}_{i=1}k_{i} \big( \left( S(v_i) S(v_{d_i}) - S(v_{d_i})S(v_{i}) \right)\omega_{d} \\
    &\quad - \left( S(v_{d_i})S(v_{i}) \right) (\omega - \omega_d ) \big) \\
    &= \sum^{n}_{i=1}k_{i} \big( S\left( S(v_i) v_{d_i}\right)\omega_{d} \big) + J(\omega - \omega_{d}) \\
    &= S(z)\omega_{d} + J(\omega - \omega_{d}).
\end{align*}

\item[5)] By \eqref{eq:ze} the norm $\|z\|^{2}=z^Tz$  is given by
\begin{align}\label{eq:N2z}
    \frac{1}{2}\|z\|^{2} 
    &= \frac{1}{2}\left( 2 e^{T}_{v} W\left( e_{0}I_{3} + S(e_{v})\right) R_{d}\right)  \notag\\
    &\quad\quad\quad \left( 2R^{T}_{d} \left( e_{0}I_{3} - S(e_{v})\right) W e_{v} \right) \notag \\
    &= 2 e^{T}_{v}W^{2}e_{v} - 2 \left( e^{T}_{v} W e_{v} \right)^{2}.
\end{align}
It is easy to see that the equality of \eqref{eq:bcond} holds when $e =[e_0, e^T_v]=[\pm 1, 0_{1\times 3}]$, which corresponds to $R=R_{d}$. 
Now, consider the undesired equilibria  $e_{j} = [0, v^{T}_{w,j}]^{T}$, being $v_{w,j}$, $j=1,2,3$, the $j$-th unit-eigenvector of $W$. Given any $ 1>\epsilon>0$,  applying an arbitrarily small  rotation $ \delta_{j} = \left[ \left( \sqrt{1-\epsilon^{2}}\right), \epsilon v^{T}_{w,j}\right]^{T}$  to  $e_{j}$ yields
\begin{equation}\label{eq:SmllR}
    e^{*}_{j} = [e^{*}_{0,j},\; e^{*T}_{v,j}]= e_{j} \otimes \delta_{j} = \left[ -\epsilon , \left( \sqrt{1-\epsilon^{2}}\right) v^{T}_{w,j}\right]^{T}.
\end{equation}
Then, the error variable $z$ in  \eqref{eq:ze} and $e_R$ in \eqref{eq:CfgErr} evaluated at $e^{*}_{j}$ and $e^{*}_{i}$ respectively, are 
\begin{align}
    \frac{1}{2}\|z^{*}\|^{2} &= 2 e^{*T}_{v,j}W^{2}e^{*}_{v,j} - 2 \left( e^{*T}_{v,j} W e^{*}_{v,j} \right)^{2} \notag\\
    &= 2\epsilon^{2}\lambda^{2}_{w,j}\left( 1-\epsilon^{2} \right) ,\label{eq:NormZ} \\
    e^{*}_{R} &= 2 e^{*T}_{v,i} W e^{*}_{v,i} = 2 \left( 1-\epsilon^{2} \right)\lambda_{w,i}. \label{eq:NormEr}
\end{align}
Therefore, $\forall i,j=1,2,3$, and 
\begin{align*}
&\forall \beta \geq 2\frac{\lambda_{w,1}}{\epsilon^{2}\lambda^{2}_{w,3}}\alpha_{1} \\
&\implies   \frac{\beta}{2}2\epsilon^{2}\lambda^{2}_{w,j}\left( 1-\epsilon^{2} \right) \geq \alpha_{1} 2 \left( 1-\epsilon^{2} \right)\lambda_{w,i} ,\\
&\implies   \frac{\beta}{2}\|z^{*}\|^{2} \geq \alpha_{1}e^{*}_{R}.
\end{align*}
Since this is the worst case for the inequality \eqref{eq:bcond} to be held,  this inequality  fulfills for all $z\neq 0_{3\times 1}$, which is equivalent to $e \in  \mathcal{S}^{3} \backslash \bigcup B_j$. 
\end{enumerate}



\ifCLASSOPTIONcaptionsoff
  \newpage
\fi



\bibliographystyle{IEEEtran}
\bibliography{VectorControl}
\end{document}